\documentclass[12pt, reqno]{amsart}
\usepackage{geometry}
\usepackage{amsthm}

\newtheorem{definition}{Definition}[section]
\newtheorem{theorem}{Theorem}[section]
\newtheorem{lemma}{Lemma}[section]
\usepackage{enumerate}
\usepackage{graphicx}
\graphicspath{{./images/}}
\usepackage[T1]{fontenc}
\usepackage[utf8]{inputenc}
\usepackage{hyperref}

\numberwithin{equation}{section}

\usepackage{listings}
\usepackage{xcolor}

\lstdefinestyle{numbers} {numbers=left, stepnumber=1, numberstyle=\tiny, numbersep=10pt}
\lstdefinestyle{MyFrame}{backgroundcolor=\color{lightgray},frame=shadowbox}

\lstdefinestyle{MyPerlStyle} {language=Perl,style=numbers,style=MyFrame,frame=lines , basicstyle=\footnotesize\ttfamily}
\lstdefinestyle{MyJavaStyle} {language=Java,style=numbers,style=MyFrame,frame=lines , basicstyle=\footnotesize\ttfamily}
\lstdefinestyle{MyPythonStyle} {language=Python,style=numbers,style=MyFrame,frame=lines , basicstyle=\footnotesize\ttfamily}

\title[Randomness from Chaos]{A Pseudo Random Number Generator From Chaos}
\author{Nabarun Mondal}
\address{D.E.Shaw \& Co. India, Hyderabad }
\email{mondal@deshaw.com}
\thanks{Dedicated to my late professor Dr. Prashanta Kumar Nandi. \\ 
Dedicated to my parents without their presence we are nothing. \\
Big thanks to: Abhishek Chanda,Gurpreet Singh (Manager),Waseem Fazal,Sahiba
Khurana, Jaipal Reddy : You all have been constant support. \\
In Memory of : Dhrubajyoti Ghosh, my closest human being. RIP 
}

\author{Partha P. Ghosh}
\address{Microsoft India, Hyderabad }
\email{parthag@microsoft.com}

\subjclass[2010]{Primary 11K45 ; 65P20 ;  Secondary 03D10 }  

\begin{document}

\keywords{
chaos ; random number generation ;  diehard ; NIST 
}

\begin{abstract}
A random number generator is proposed based on a 
theorem about existence of chaos in fixed point iteration of $x_n =  cot^2 ( x_{n-1} )$. 
Digital computer simulation of this function iteration exhibits random behavior.
A method is proposed to extract random bytes from this simulation.  
Diehard and NIST test suite for randomness detection is run on this bytes, 
and it is found to pass all the tests in the suite. 
Thus, this method qualifies even for cryptographic quality random number generation.  
\end{abstract}

\maketitle

\begin{section}{Introduction}\label{intro}

Fast generation of high quality random numbers using simple arithmetic methods are elusive.
The most elusive aspect about generating random numbers is that 
one can never be sure, if the produced numbers are \emph{really random}. 
That is exactly why a background theory is needed for it, and a set of \emph{standardized} statistical tests
must \emph{certify the numbers as random}.

Robert R. Coveyou suggested : \emph{``The generation of random numbers is too important to be left to chance.''}. 
Almost same theme is reiterated in the immortal words of Knuth \cite{DK} 
\emph { ``one should not expect arbitrary algorithms to produce random numbers, a theory should be involved'' }.
The great Von Neumann asserted randomness lies beyond the realm of arithmetic: 
\emph{``Any one who considers arithmetical methods of producing random digits is, of course, in a state of sin.''}   
As it turned out to be, Von Neumann might not had been exactly correct, and that is what the authors want to present in the current paper.

Until recently it was assumed that randomness comes out of complexity. 
However, it can also be an effect of chaotic dynamics \cite{jj}.
Properties of a chaotic system would be discussed in the section \ref{theory}. 
Many people misleadingly talk about as if there are
various sources of randomness. According to some authors  \cite{jj} , the fundamental source of randomness, 
if at all any, is Heisenberg Uncertainty Principle.

A practical viewpoint is : \emph { Randomness occurs to the extent that something can not be predicted.}
Randomness is a matter of degree, and lies in the eyes of the beholder.
Poincare pointed out that the classic random outcome of a die throwing or a flipping coin, comes from
sensitive dependence to the initial condition. A small perturbation causes a large difference in the final  outcome,
thereby making prediction difficult. This sensitive dependence is a hallmark of \emph{Chaos} \cite{cfnfs}.

Chaotic trajectories even look random, and they pass many classic \emph{``tests'' of randomness}. This in fact generates the \emph{principle of equivalence} between
chaotic and random systems, as discussed in \cite{CW}.

Furthermore, chaotic systems might arise in very simple forms of iterative maps \cite{csi} \cite{cfnfs} (definition \ref{ifs}). 
Two such \emph{simple looking} systems are discussed in the section \ref{theory}, both exhibiting complex behavior.
They form the operating principle (or the theory) behind the proposed random number generator is presented.

However, chaotic systems trajectories sometimes can approach 
the subset of the state space called \emph{attractors},
thereby drastically reducing the perceived randomness.  
This is a natural outcome of running the simulation in digital computer, 
where the precision of the state point is limited to fixed number of bits of information,
without arbitrary precision, the simulated chaotic system would get into a cycle.  

The practical principle of \emph { garnering randomness out of chaotic systems } would be then, 
to ensure, that the system trajectory does not end up being in an attractor. 
Practical generation of random digits (bytes) are the topic of discussion
in section \ref{expt}. 

Finally, the result of two of the industry standard statistical test suites are presented in section \ref{results}.
It was found that the generated random byte stream passes
all these industry standard tests, thereby making the generated numbers suitable even for cryptography. 

\end{section}

\begin{section}{Theory of Operation}\label{theory}

In this section we would discuss the operating principle of the proposed
random number generator, and prove that, the iterative map is \emph{Devany 
Chaotic} (definition \ref{chaos} \cite{devany} ). Some of the prerequisite terms and definitions 
can be found in Appendix \ref{ap_1}.

We start with the idea of a fixed point iteration, 
or a single dimensional discrete dynamical system. 

Let there be a function $f : X \to X$ where $X$ is a metric space (definition \ref{mp}). 
Lets start with $x_0 \in X$, and let $x_1 = f(x_0)$. 
Let $x_2 = f(x_1)$. 

\begin{definition}\label{ifs}
\textbf{Iterative Function.}

The general form :-
$$
x_{n+1} = f ( x_n )
$$
is called iterative function system, or a fixed point iteration,
which is a type of single  dimensional discrete dynamical system, or map. 
\end{definition}

Many iterated systems exhibit what is known as ``chaotic behavior''.
``Chaos'' however remains a tricky thing to define. 
In fact, it is much easier to list properties that a system described as ``chaotic'' has 
rather than to give a precise definition of chaos \cite{devany}.

\begin{definition}\label{chaos}
\textbf{Chaotic Systems.}

A chaotic fixed point iteration (dynamical system \ref{ifs} ) 
is generally characterised by:-
\begin{enumerate}
\item{Having a dense (definition \ref{dense-set}) collection of points with 
(definition \ref{orbit}) periodic orbits.}
\item{Being sensitive to the initial condition of the system (so that initially nearby points can evolve quickly into very different states), a property sometimes known as the butterfly effect.} 
\item{Being topologically transitive (definition \ref{top-trans}).}
\end{enumerate}

\end{definition}

\begin{subsection}{Discontinuous Maps}
The map we would be discussing here is $x \leftarrow cot^2(x)$. 
This is a  discontinuous map. Unfortunately, while literature is filled with material on continuous maps, very few are really available on discontinuous maps \cite{SC}.
Reciprocal of any continuous map becomes discontinuous on the set of zeros
of the original map.

For example the map $x \leftarrow 1/x^2$ is a reciprocal map of $x^2$,
while the zero of the original $x=0$, became a point of discontinuity.

In general there is nothing interesting about these type of reciprocal maps, 
as this specific one has a fixed point $x^* = \pm 1$, and while iterated 
would immediately converge to one of them.

But, these maps behave very differently when put into a circle domain,
that is in modular form.

Such a modular form for the map $x \leftarrow 1/x^2$ is:-
$$
x_n = \left ( \frac{1}{x_{n-1}^2} \right) \; mod \; 1 \; ; \; x \in (0,1) 
$$ 
For some maps which are inherently periodic, this contraption becomes 
a natural choice to investigate their behaviour. 

For example we can define the reciprocal map $x \leftarrow 1/cos(x_{n-1})$,
whose domain is :- 
$$
D = \mathbb{R} \setminus \{ (2n +1)\pi/2 \}  \;  ; \; n \in \mathbb{Z}   
$$
But what happens in real is modular arithmetic with period $\pi$.
Hence, any output $x_n > \pi$ or $x_n < 0 $ 
should be treated modulo $\pi$, due to periodicity of $cos(x)$.
In this case, the modified system have the domain changed from $D$ to :-
$$
D_{\pi} = [0, \pi/2 ) \cup ( \pi/2, \pi]
$$ 
The output range should also be changed accordingly to :-
$$
x_n = \frac{1}{cos(x_{n-1}) } \; mod \; \pi 
$$ 
This works in a domain where $0,\pi$ are actually the same point.
Notice the unnecessary use of the period $\pi$ in the definition.
Therefore, we can generalise this for any periodic function 
with a period $T$ and we can go a step beyond to always normalise
the domain and range as the same interval: a.k.a the unit circle :-
$$
I = (0,1)
$$ 
by introducing the following definition of normal maps:-
 
\begin{definition}\label{nf}
\textbf{ Normalised Maps. }

Let $g(x)$ be any function with a period $T$ .
Then:-
\begin{equation}\label{nfe}
N_g(x) = \frac{ g(xT) \; mod \; T}{T} \; ; \; x \in (0,1)  
\end{equation}
is called the normalised rational map form of $g(x)$.
\end{definition}

With this definition is mind, lets see how a modulus operation applies on 
a linear map $y(x) = mx + C$ with mod $T$.
We note that obviously when $y_n = T$ , $y_n$ would become zero,
and would generate discontinuity, but while $y_n < T$, the derivative
remains $y' = m$. The modulus operation does not change the derivative
of the underlying function, just makes it undefined at periodic
points $nT \; ; \; n \in \mathbb{Z}$. 

Comparing equation \eqref{nfe} of definition \ref{nf} with 
the sawtooth map equation \eqref{st} 
suggests that the map $y(x)$  would become a sawtooth type map, 
which is known to be of chaotic type,
given the slope is greater than 1, that is $|m| > 1$.

\begin{equation}\label{st}
f(x) = m ( 1  - x ) \;  mod \; 1 \; ; \; m \in \mathbb{R} 
\end{equation}

In case of a general function $g(x)$, only the slope would vary,
which would become a function of x itself as in :-
$$
m(x) = g'(x) .
$$ 

A normalized function (definition \ref{nf}) 
diverging to infinity (having singularities)
would still have the derivative of its non modular analog, 
but would have introduced countable 
discontinuities due to  modular operation.

Image \eqref{inv_x_sq_normal} demonstrate the effect modularisation has on the 
function $1/x^2$. Note the discontinuities appearing on the side of $x \to 0$,
they are all chopped off to $(0,1)$, and as we are mapping a potentially 
infinite length with unit interval, we would need countable amount of them.  

The fixed point $x^* = 1$ can not be reached from anywhere in the domain now. 
If an iteration value $x_n$ gets close to the value $1$,
in the next iteration it would be pushed back to the left side of the map, 
with $x_{n+1} \approx 0$. 
Which means it would go into the zone of the chaos, 
where it will be eventually pushed back 
to any point other than 1. This dynamics is what essentially generates 
the chaotic behaviour in the system.
These type of systems have tremendous sensitivity to initial condition.

\begin{figure}[htbp]
  \centering
  \fbox{
    \includegraphics[scale=0.4]{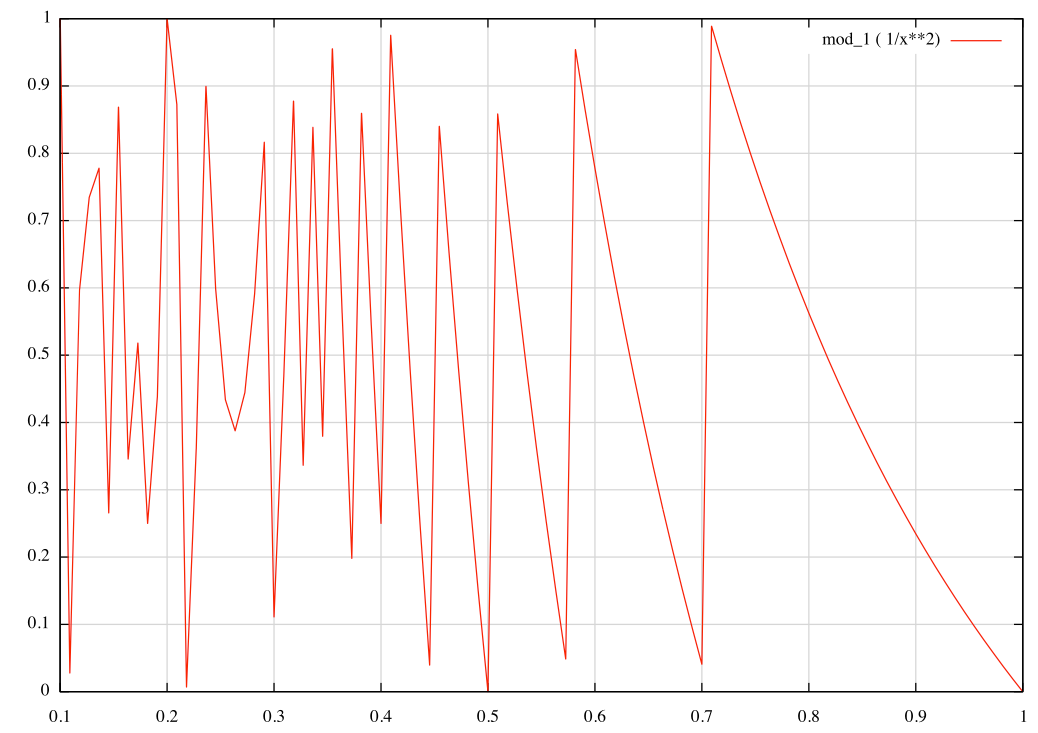}
  }
  \caption{ Function :  $1/x^2$  mod 1, showing the sawtooths }
  \label{inv_x_sq_normal}
\end{figure}

In fact, the modularisation  makes many trivial map
act in a non trivial way, for example the always 3-cycle map
of $x \leftarrow 1/(1-x)$ under the normal form 
$$
x_n = \frac{1}{1 - x} \; mod \; 1 \; ; \; x \in (0,1)
$$ 
becomes just as chaotic as depicted in Image \ref{inv_x_sq_normal}.

In this regard, it is curious to note how modulus operation 
can generate chaotic behaviour. Modulus operation
is a very special case of the general discontinuous maps $\mathcal{C}$  
studied in \cite{SC}, which we define next.

\begin{definition}\label{scf}
\textbf{Sharkovosky Chua Type Map \cite{SC} : $f \in \mathcal{C}$. }

Let $f : I \to I $ with $[a,b]$ be map such that:-
\begin{enumerate}
\item{$f$ is continuous everywhere except points :-
$$
D = \{ z_1 , z_2 , ..., z_r\} \; ; \; r \in \mathbb{Z} 
$$
}
\item{ $f$ is monotonic in the interval $(z_i,z_{i+1})$ for $i = \{ 0, 1, ..., r \}$,
with $z_0 = a, z_{r+1} = b$.
}
\item{ The limits 
$$
\lim_{x \to z_i - 0 }{f(x)} \ne \lim_{x \to z_i + 0 }{f(x)}
$$ and
$$
\lim_{x \to z_i \pm 0 } = f(x) \in \{a,b\} \forall i
$$ 
}
\item{ It is expansive $|f'(x)|= l >1$. 
For every point $x \in I \setminus \{z_i\}$ 
there exists an interval $U_x$ containing the point $x$ such that:-
$$
d(f(U)) > l d(U)
$$  for every interval $U \subset U_x$ 
where $d(V)$ is the length of the interval $V$.
}
\end{enumerate}
then, $f \in \mathcal{C}$ is called a Sharkovosky Chua Type Map. 
\end{definition}

Note that the original papers \cite{SC} definition does  
admit countable number of infinities, but that admission 
does not change any property of the map with respect to chaotic behaviour,
in fact augments it.

Now, we would introduce a special class of maps with more general
properties then Sharkovosky Chua type maps, 
these maps allow a specific type of non expansive intervals in the function domain.   

\begin{definition}\label{scg}
\textbf{Generalized Sharkovosky Chua Type Map : $f \in \mathcal{G}$.}

A function $f$ is said to be of Generalised Sharkovosky Chua type 
$f \in \mathcal{G}$, iff the function has the properties 
from (2,3) and modified (1, 4) with the following:-
\begin{enumerate}
\item{ Countable number of discontinuity. }
\item{ Every element  $x_s \in I_s$ such that $|f'(x_s)| < 1 $, 
has an element in forward orbit $x_g$ such that $|f'(x_g)| > 1 $.
}
\item{ It is expansive $|f'(x)|= l >1$. 
For every point $x \in I \setminus  (\{z_i\} \cup \{ I_s \}  ) $ 
there exists an interval $U_x$ containing the point $x$ such that:-
$$
d(f(U)) > l d(U)
$$  for every interval $U \subset U_x$ 
where $d(V)$ is the length of the interval $V$.
}

\end{enumerate}
\end{definition}

\begin{lemma}\label{reduction}
\textbf{Reduction to Sharkovosky Chua Type Map.}

A Generalised Sharkovosky Chua Type map (definition \ref{scg} ) 
can be \emph{contracted} into a Sharkovosky Chua Type map (definition \ref{scf}). 
\end{lemma}
\begin{proof}
The reduction is straight forward.

Informally, we contract the orbits from an expansive region to 
a non expansive region, to another expansive region.

As the intervals $\{I_s\}$ forwards all the points to 
the points $x_g$ such that $|f(x_g)| > 1$, we can contract
those iterative steps (the orbit from $x_{g'}\to x_s \to x_g$ ) 
into a single step, and eliminate the
whole intervals $\{I_s\}$ where $|f(x_s)| <1$, thereby contracting 
the intervals, and no point $x_s$ remains where $|f'(x_s)| < 1$.

\end{proof}

\begin{theorem}\label{chaos-sc}
\textbf{ Sharkovosky Chua Type Maps are Chaotic.}

Sharkovosky Chua Type Maps are chaotic, 
so are the generalised Sharkovosky Chua Type maps. 
\end{theorem}
\begin{proof}
This has been discussed in \cite{SC}.
As the generalised Sharkovosky Chua maps can be contracted 
into a  Sharkovosky Chua map, they are chaotic too. 
\end{proof}

Now we give one example of a normalised map that is not in $\mathcal{G}$.
Curiously , the map $1/sin^2(x) \not \in \mathcal{G}$, because 
the non expansive zone does not push the orbit back to the chaotic zone,
as required by definition \eqref{scg} property (3).
As the fixed point (about 0.37) lies in the non expansive zone, 
the map converges, and is not chaotic at all. 

This can nicely be explained with the Image \ref{inv_sin_sq_normal},
see at the zone where the function has zero, does not have the expansive
property, the zone is stable.

\begin{figure}[htbp]
  \centering
  \fbox{
    \includegraphics[scale=0.4]{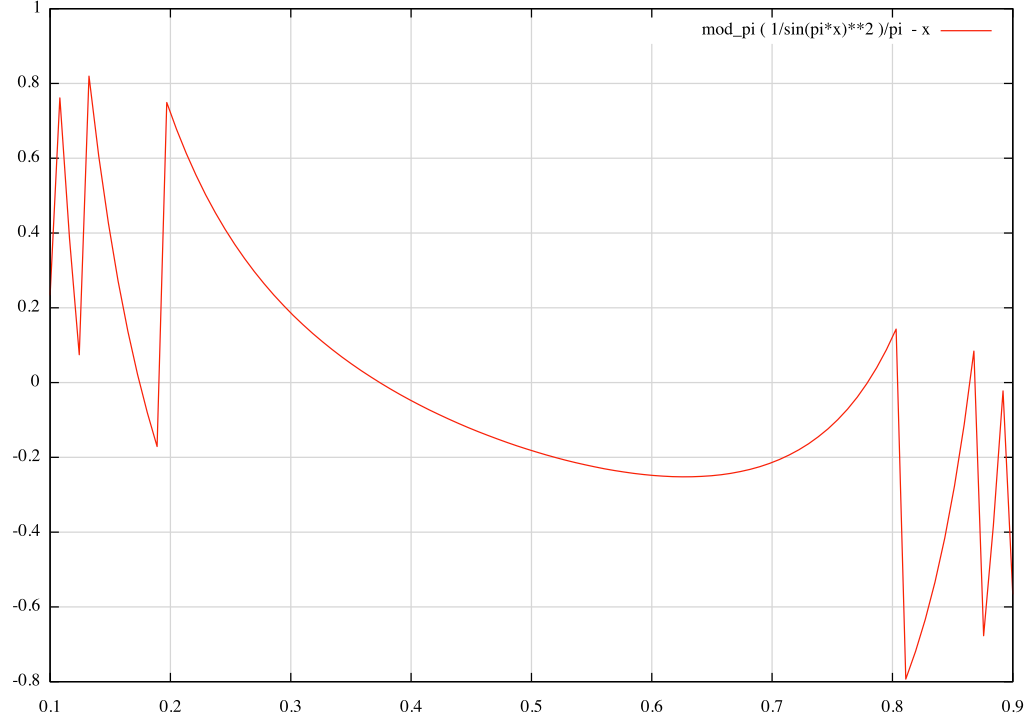}
  }
  \caption{Normal form of [ $1/sin^2(x)$ - x ] showing the fixed point zone }
  \label{inv_sin_sq_normal}
\end{figure}

What maps are of Generalised Sharkovosky Chua type $f \in \mathcal{G}$ ?
Take a polynomial map with singularity, normalise it, it probably goes in $\mathcal{G}$.
Examples are $1/x^2$ map, $1/(1-x)$ map. 
The trigonometric reciprocal maps $1/cos(x) \in \mathcal{G}$, 
so is the map $1/cos^2(x)$
The map $tan(x) \in \mathcal{G}$ and $cot(x) \in \mathcal{G}$.

Finally, the map $cot^2(x) \in \mathcal{G}$, 
which is discussed in the next subsection.

\end{subsection}

\begin{subsection}{Reciprocal Cot Squared Map.}

In this subsection we discuss the dynamics of the $cot^{-2}(x)$ map.
This is formally represented by equation \eqref{cot-2}
\begin{equation}\label{cot-2}
x_n = \frac{1}{cot^2(x_{n-1})} =  \left ( \frac{cos(x_{n-1})}{sin(x_{n-1})} \right)^2
\end{equation}

This map can is presented in the semi normal form,
that is the domain as $(0,1.0)$ but the range is in $\mathbb{R}$ 
in the figure \eqref{cot_2}. In the figure \eqref{cot_2}, we did not 
present the singular points $\{0,1\}$ and draw only the range $(0.1,0.9)$.

\begin{figure}[htbp]
  \centering
  \fbox{
    \includegraphics[scale=0.4]{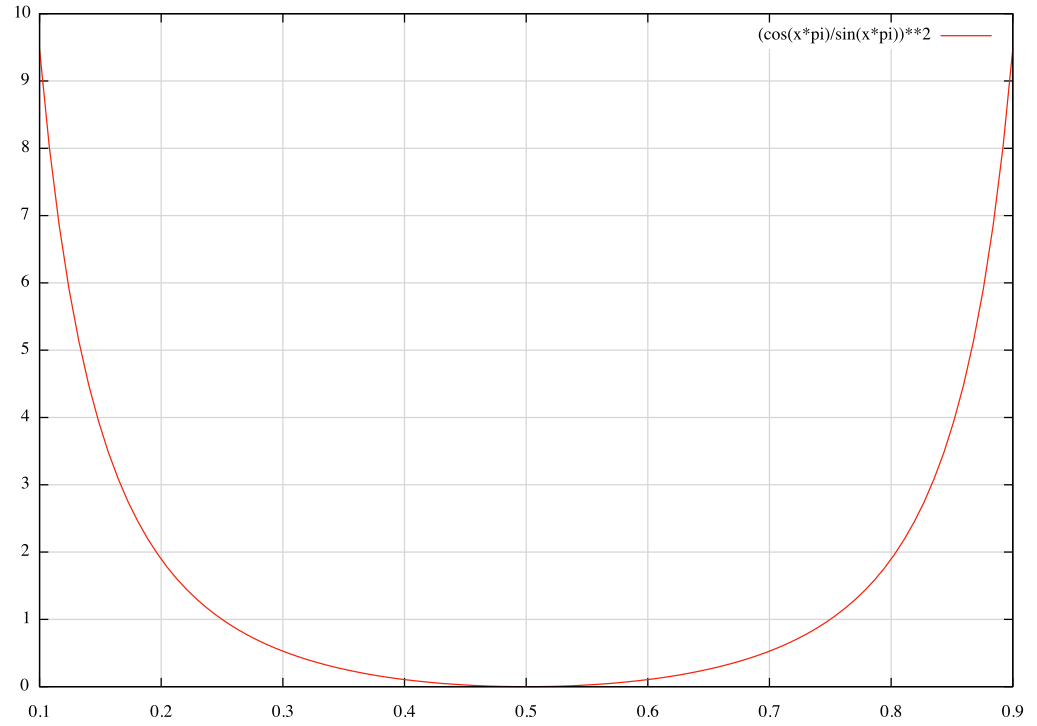}
  }
  \caption{The function  $cot^2(x \pi)$ }
  \label{cot_2}
\end{figure}

It is apparent that due to symmetry, to normalise the function 
it can be taken as in $\pi/2$ instead of $\pi$, as if we normalise with period $\pi$,
the normalised function $N_c$ would be $N_c(x) = N_c(1-x)$.

The function normalisation can be seen in figure \eqref{cot_2_n}. 

\begin{figure}[htbp]
  \centering
  \fbox{
    \includegraphics[scale=0.35]{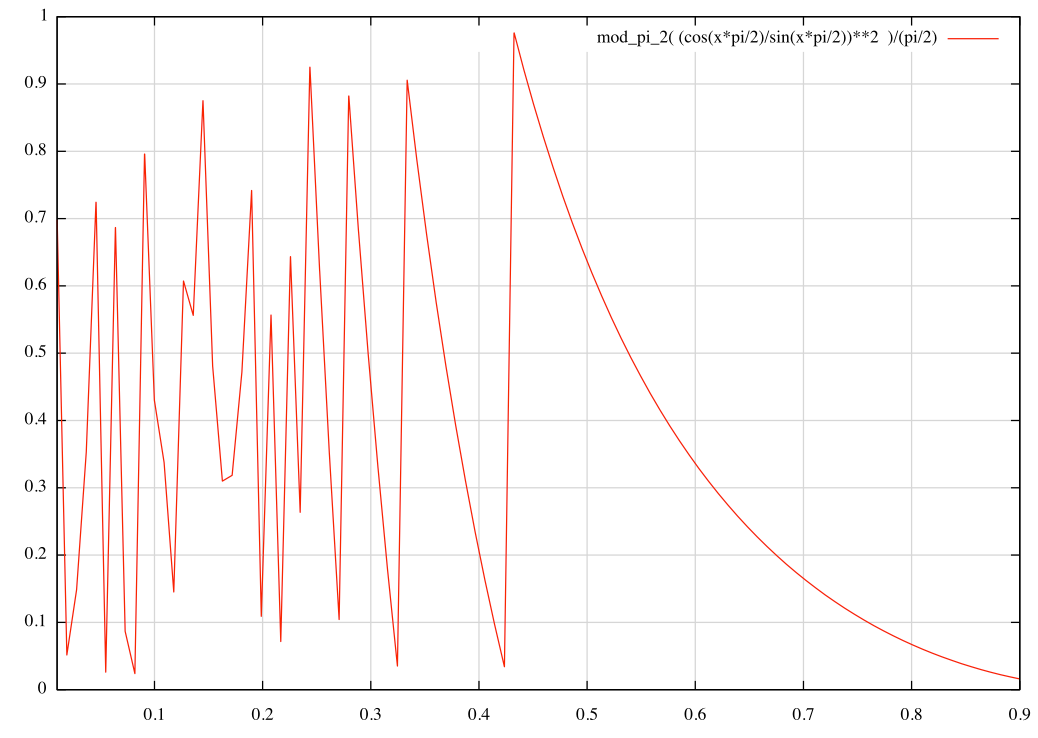}
  }
  \caption{Normal form of $cot^2(x)$ from equation \eqref{normal-cot} }
  \label{cot_2_n}
\end{figure}

\begin{figure}[htbp]
  \centering
  \fbox{
    \includegraphics[scale=0.35]{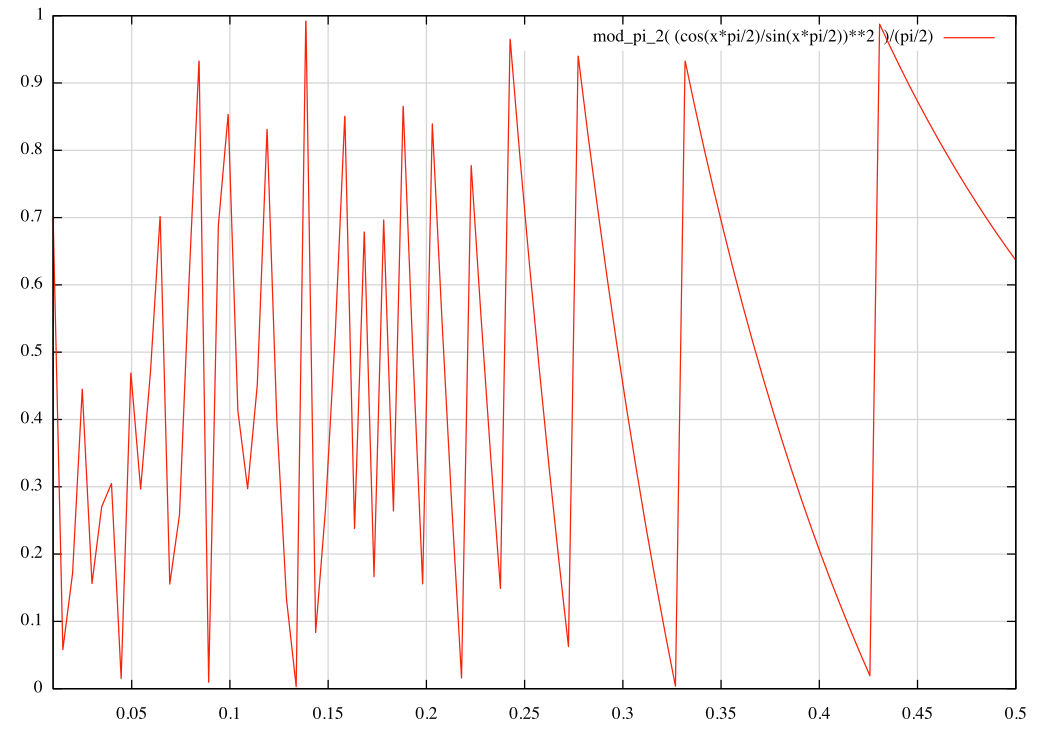}
  }
  \caption{ The zone of chaos in $N_c$ (equation \ref{normal-cot}) }
  \label{cot_2_n_c}
\end{figure}

\begin{figure}[htbp]
  \centering
  \fbox{
    \includegraphics[scale=0.35]{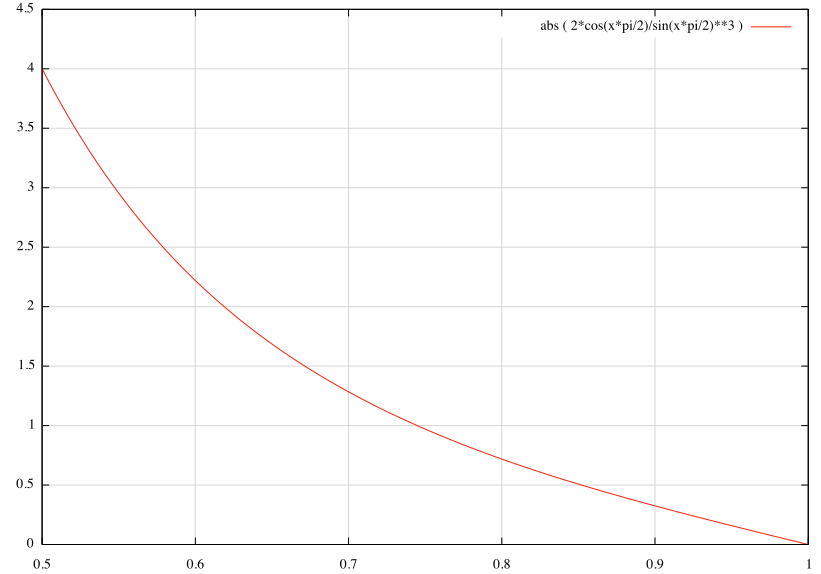}
  }
  \caption{ The Normalised Derivative $cot^2(x)$ \eqref{der} }
  \label{cot_2_n_d}
\end{figure}

\begin{theorem}\label{cot-sq-c} 
\textbf {Iteration of $x_{n+1} = cot^2 ( x_n ) $ is chaotic.}

The iteration of $x_{n+1} = cot^2 ( x_n ) $ 
by being a generalised Sharkovosky Chua type map \ref{chaos-sc} is chaotic. 
\end{theorem}
\begin{proof}
The normalised function (definition \ref{nf}) 
form for $cot^2(x)$ with period $\pi/2$  can be written by equation \eqref{normal-cot}.

\begin{equation}\label{normal-cot}
N_c(x) = \frac{2}{\pi} \left (  cot^2 \left ( \frac{\pi x }{2} \right) \; \; mod \; \frac{\pi}{2} \right )     
\end{equation}

It is apparent from the figure \eqref{cot_2_n} 
that for values  of $x_n \to 1 $ 
the next value $x_{n+1}$ would be in the chaotic zone, 
depicted in figure \eqref{cot_2_n_c}.  

The derivative is :-
\begin{equation}\label{der}
\left |\frac{d[cot^2(x)]}{dx} \right| = \left  | \frac{2cos(x)}{sin^3(x)} \right | 
\end{equation}

We note that the derivative $|f'(x)| < 1 $ only when $x_n > 0.7$
however, at that point the next $x_{n+1}$ would push it between $(0.1,0.3)$
which is a massively chaotic zone as depicted by figure \eqref{cot_2_n_c}.
Hence, every non expansive point maps back into the zone of chaos,
which is also the expansive zone with $|f'(x)|>1$ .

Hence equation \eqref{normal-cot} satisfies condition of generalised 
Sharkovosky Chua type map, and hence is chaotic.
\end{proof}

We now present the theorem of \emph{principle of equivalence between
chaotic and random behavior}, which would be used to apply the 
random behaviour of the iterate $cot^2(x)$.
 
In a very insightful paper \cite{CW} it has been proven that:- 
\emph{chaotic and random systems are observationally indistinguishable}.
If that is the case, then, 
one can replace a random system by an equivalent chaotic system, 
and vice versa, as has been argued in \cite{CW}. 

\begin{theorem}\label{poe}
\textbf{Chaotic systems and Random Systems are Observationally Equivalent
\cite{CW}.
}

A chaotic system can be observationally replaced with a random system, 
and vice versa. 
\end{theorem} 

\end{subsection}
\end{section}

\begin{section}{Experimental Validations }\label{expt}
As the theorem \ref{poe} suggests, by virtue of being chaotic, 
the iteration of theorem \eqref{cot-sq-c} could serve as a randomness generator.
However, there are cautionary words involved because 
\emph { arbitrary precision } is needed while doing so.  

When the digits are of finite precision, 
say 2 digits before and after decimal, 
the system can not exhibit an arbitrary long period, 
as the total number of states possible are $10^{2+2} = 10^4$.   
In general if the number of digits in base $b$ are of $N$ , 
then the maximum possible number of states $|S|$ becomes:-
$|S| = b^N $ which is still finite. 
However, as we keep on increasing $N$, $|S|$ increases exponentially,
and the chance of repeat decreases.

The perl code for computing the iteration of $x = cot^2(x)$ is presented here.
\begin{lstlisting}[style=MyPerlStyle]
sub cot
{
	my $x = $_[0];
	cos($x) /sin ($x) ;
} 
while ( 1 )
{
	printf  "%.18f\n", $x ;
	$x = cot($x) ** 2 ;
}
\end{lstlisting}
The above code generates $x \in \mathbb{R}_+$,
which is not well suited for visualizing the generated randomness.
A \emph{chaotic dense orbit} can 
only be visualized in case  of a bounded domain. 
To move the resulting value in the bounded $(0,1)$ domain the following 
code can be used:-
\begin{lstlisting}[style=MyPerlStyle]
 while ( 1 )
 {
	$x = cot($x) ** 2 ;
	my $fract = $x - int ( $x ) ;
 	printf "%.18f\n", $fract ;
 }
\end{lstlisting}

This code generates $x \in (0,1)$.
The result of the iteration is \emph{dense} (definition \ref{dense-set}) 
as depicted in the figure \ref{cot_2_e},
and a close-up view of the iteration is shown in the figure \ref{cot_2_close}.

\begin{figure}[htbp]
  \centering
  \fbox{
    \includegraphics[scale=0.6]{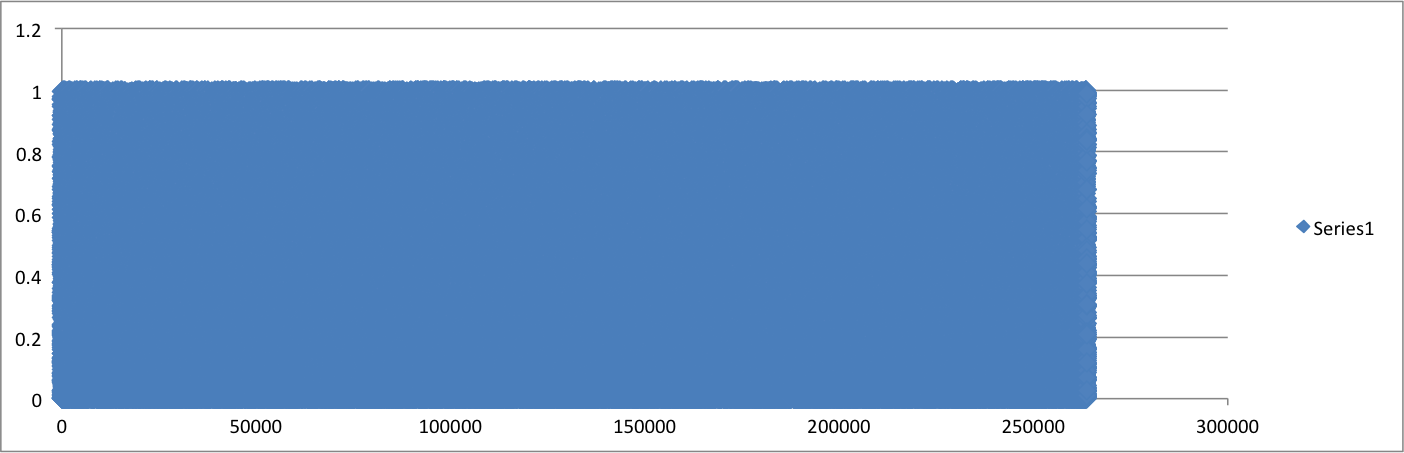}
  }
  \caption{Iteration of $x_{n+1} = cot^2(x_n)$ (fractional part)}
  \label{cot_2_e}
\end{figure}

\begin{figure}[htbp]
  \centering
  \fbox{
    \includegraphics[scale=0.7]{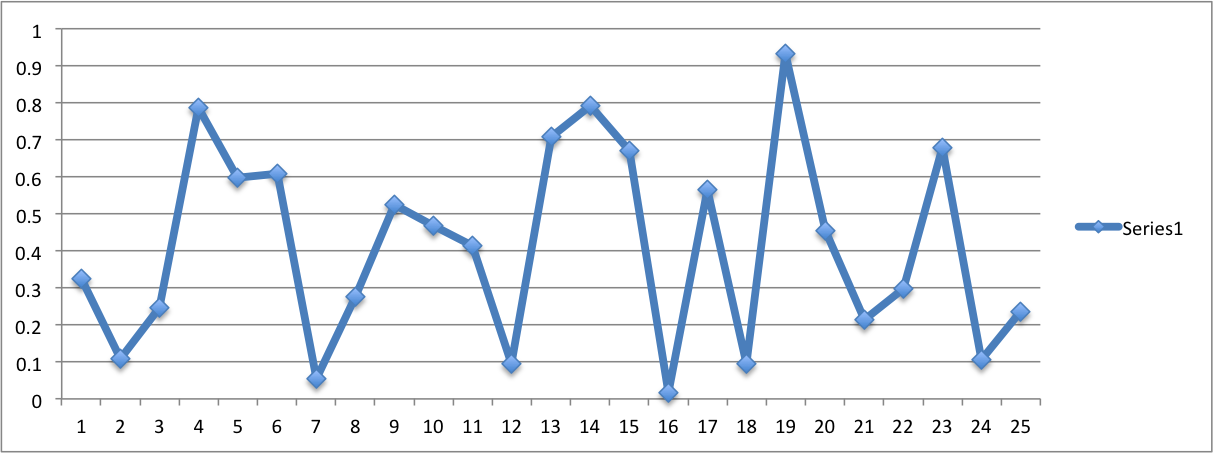}
  }
  \caption{Iteration of $x_{n+1} = cot^2(x_n)$ Showing a small Range (fractional part)}
  \label{cot_2_close}
\end{figure}

However, all industry standard statistical tests run on binary byte streams of data.
Therefore to \emph{test for randomness} one has to transform 
this sort of random ``real'' data points into a stream of random bits.

Hence, we needed to convert this \emph{apparently random iteration} values 
into stream of bytes so that randomness test suites like diehard/NIST 
can be used to verify the randomness. 
This is discussed in the next subsection.

\begin{subsection}{Digital Representation of Real}\label{ieee-d}

The IEEE floating point specification mandates that the \emph{double} representation would be 
of 64 bits or 8 bytes. If the representation starts with bit number $0$ as $b_0$ and ends with bit number
$63$ as $b_{63}$ , then, the $b_{62}..b_{52}$ is exponent of the number, while $b_{63}$ is the sign bit.
In byte wise notation it is $B_7B_6...B_1B_0$ .

So, if the function $x_{n+1} = cot^2(x_n)$ is really chaos generating, then, if we can gather 
the mantissa : $b_{51}...b_0$ of the $x_n$ , the resulting bitstream should behave as random.
However, it wold not be optimal for extracting $b_{51}...b_0$ , as they are not in the byte boundaries.
A better, faster option would be discarding 2 of the most significant bytes, that is $B_7 B_6$ and taking 
the other bytes $B_5 B_4 ... B_0$ as random bytes.

The Python language uses $hex()$ call to show the floating point number :-
$$
[sign] \; [0x] \; integer \; [. \; fraction]\;  [p \; exponent]
$$
Hence, the following python code demonstrates the exact random bytes as in the hexadecimal string.
\begin{lstlisting}[style=MyPythonStyle]
import math
import time
x = time.time() #Seed it, go to Chaos.
x = x - math.floor( x ) 
while True:
    y = math.tan(x)
    y = y * y
    x = 1.0/y
    s = x.hex()
    a = s.split(".")
    b = a[1].split("p")
    s = b[0]
    print s
\end{lstlisting}

The resulting output is shown below:-
\begin{lstlisting}[style=MyPythonStyle]
49c07b03923a2
2201d34423ec8
47527dadaecb0
\end{lstlisting}

The following Java class code precisely accomplishes the same thing:-

\begin{lstlisting}[style=MyJavaStyle]
public class DoubleImpl  {
    private double number;
    public DoubleImpl()
    {
        number = System.nanoTime() ;
    }
    public double getNumber()
    {
        double r = Math.tan(number) ; 
        number = 1/(r*r) ;
        return number;
    }
    public byte[] getRandomBytes()
    {
        long l = Double.doubleToRawLongBits(getNumber()); 
        byte[] bytes = toByta_StripSignByte(l);
        return bytes;
    }
    /* Return bytes from B_5 to B_0  here */
    public static byte[] toByta_StripSignByte(long data) {
        return new byte[] {
            (byte)((data >> 40) & 0xff),   //B_5
            (byte)((data >> 32) & 0xff),   //B_4
            (byte)((data >> 24) & 0xff),   //B_3
            (byte)((data >> 16) & 0xff),   //B_2
            (byte)((data >> 8) & 0xff),    //B_1
            (byte)((data >> 0) & 0xff)     //B_0
        };
    }
} 
\end{lstlisting}
The above code is platform independent, and authors also wrote a faster C implementation.
\end{subsection}
\end{section}

\begin{section}{Results of The Statistical Testing}\label{results}
For statistical testing suites a binary file is to be passed
to the test suite, which is then read as a stream of random bits.

We called the \emph{getRandomBytes} function (which generates 6 random byte each time it gets called) to generate medium (min:600 Mega Bytes) to long ( Max: 6 Giga Bytes) binary files. 

On a 2.66 GHZ Core 2 Duo Macbook pro, with OS X Lion, using the \emph{C langauge} variant , the 60 Megabyte file takes 2 sec, 
for 600 MB 19 sec, and 6 Giga Bytes of random bytes generation takes 290 secs.     

This file, was passed to the diehard family \cite{DIEHARD} of test suite.     
The findings are summarised  in the  table [\ref{table:pr-diehard-results}].

However, diehard is an old generation of test suite, and NIST suite \cite{NIST}
is the current industry norm for randomness testing.
This generator passes all the NIST suite \cite{NIST} tests with ease.
The file size of 6 gigabyte was passed to the NIST suite, with minimum bitstream length of 10000 bits, 
with minimum number of bitstream 1000, to a maximum of 10000. 
The resulting p-values are aggregated over all the bitstreams.

The results are summarized in the table  [\ref{table:pr-nist-results}].

\begin{table}[!ht]
\caption{Results of the Diehard Suite }
\centering
\begin{tabular}{ | l | l |  l |  }
\hline\hline
Test Name                             & Resulting p-value[s]       &  End Result  \\ [0.5ex]
\hline
Birthday Spacing                      & 0.948788                                & PASSED     \\ \hline
Overlapping 5-Permutation             & 0.874604,0.993319                       & PASSED     \\ \hline
Binary Rank Test (31X31) Matrix       & 0.724738                                & PASSED     \\ \hline
Binary Rank Test (32X32) Matrix       & 0.853068                                & PASSED     \\ \hline
Binary Rank Test (6X8) Matrix         & 0.028129                                & PASSED     \\ \hline
The Bit Stream Test                   & min 0.0545,  max 0.8906                 & PASSED     \\ \hline
OPSO                                  & min 0.0373,  max 0.9579                 & PASSED     \\ \hline   
OQSO                                  & min 0.0500,  max 0.9688                 & PASSED     \\ \hline  
DNA                                   & min 0.0359,  max 0.9986                 & PASSED     \\ \hline    
Count The 1's  On Byte Stream         & 0.223283, 0.165772                      & PASSED     \\ \hline
Count The 1's  On Specific Bytes      & min 0.022031,max 0.967922               & PASSED     \\ \hline
Parking Lot                           & 0.343578                                & PASSED     \\ \hline
Minimum Distance Test                 & 0.097679                                & PASSED     \\ \hline
3D Sphere Test                        & 0.885174                                & PASSED     \\ \hline
Squeeze Test                          & 0.805752                                & PASSED     \\ \hline  
Overlapping Sums Test                 & 0.531035                                & PASSED     \\ \hline
Runs Test Up                          & 0.297245, 0.766326                      & PASSED     \\ \hline
Runs Test Down                        & 0.315770, 0.857272                      & PASSED     \\ \hline
Craps Test                            & wins:0.360479,throws/game:0.207756      & PASSED     \\ \hline

\hline
\end{tabular}
\label{table:pr-diehard-results}
\end{table}

\begin{table}[!ht]
\caption{Results of the NIST Suite }
\centering
\begin{tabular}{ | l | l |  l |  }
\hline\hline
Test Name                             & Resulting p-value[s]       &  End Result  \\ [0.5ex]
\hline
Frequency                             & 0.379045                                & PASSED     \\ \hline
Block Frequency                       & 0.963497                                & PASSED     \\ \hline
Cumulative Sums (2 tests)             & 0.660844,0.895204                       & PASSED     \\ \hline
Runs                                  & 0.013760                                & PASSED     \\ \hline
Longest Run                           & 0.928150                                & PASSED     \\ \hline
Rank                                  & 0.007694                                & PASSED     \\ \hline
FFT                                   & 0.081013                                & PASSED     \\ \hline   
NonOverlappingTemplate (multiple)     & -                                       & PASSED     \\ \hline  
Universal                             & 0.678686                                & PASSED     \\ \hline    
Approximate Entropy                   & 0.048404                                & PASSED     \\ \hline
RandomExcursions (multiple)           & -                                       & PASSED     \\ \hline
RandomExcursionsVariant (multiple)    & -                                       & PASSED     \\ \hline
Serial (2 tests)                      & 0.019966, 0.175884                      & PASSED     \\ \hline
LinearComplexity                      & 0.679514                                & PASSED     \\ \hline

\hline
\end{tabular}
\label{table:pr-nist-results}
\end{table}

\end{section}

\begin{section}{Usage Advantages}\label{advantages}
The premise of the operation of the $cot^2$ generator is the theorem \ref{cot-sq-c}.
Due to the nature of the iteration (chaotic), the following properties are true :-

\begin{enumerate}

\item{\textbf{Hard to Predict.}  

The seed of the generator is $x_0$, which can be an arbitrary precision floating point value.
For the orbits of the iteration with $x_0 = 1.0000000$ and with $x_0 = 1.0000001$ 
the resulting byte stream start differing from the 3rd byte.
The choice of the default initial value $x_0$ as current systems time in nano second, 
makes it impossible to track and predict the value for any iteration,
unless absolutely precise time synchronisation is achieved between two systems.  

Due to machines architectural and implementation differences,
the same code produces two completely different byte streams in \emph{C++} and \emph{Java}.
Unless one knows the exact architecture of the machine and environment (language of implementation) on which the algorithm is running, 
it would be hard to predict the outcomes. 
}

\item{\textbf{ Conjectured to be Unique. }

For sufficient precision , 
the value of the $x_n$ will almost never repeat. The machine simulation shows this to be true for default precision, up to millions of iterations. 
Given the whole system can be made using arbitrary 
precision floating point arithmetic libraries of choice, then the values would be 
almost surely unique.
}

\item{\textbf{ Fast Algorithm. }

The java implementation for this random generator  
is faster than 
$$
java.security.SecureRandom
$$ 
implementation.
To generate a 600 mega byte of random bytes, using 6 bytes of buffer, 
current java implementation takes 150 seconds on a SUN Ultra-SPARC machine, 
running java 7 on SOLARIS,  while SecureRandom takes more than 230 Seconds.
}

\end{enumerate}

\end{section}

\begin{section}{Summary and Future Works}\label{summary}

In this paper we have demonstrated a computationally easy and fast method of 
generating very good random numbers for all practical purposes, including cryptography. 
Arbitrary large random chunk of bytes can be generated if we take an arbitrary precision Real Number implementation like BigDecimal in Java, albeit the
number generation would get slower. 
On the flip side, as the principle is well established, 
we can guarantee arbitrarily large period for the random number generator.
\end{section}

\appendix

\begin{section}{Definitions of The Theory Section.}\label{ap_1}
\begin{definition}\label{fp}
\textbf{Fixed Point of a function. }

For a function $f:X \to X$ , $x^*$ is said to be a fixed point, iff $f(x^*) = x^*$ .
\end{definition}

\begin{definition}\label{mp}
\textbf{Metric Space.}

A metric space is an ordered pair $(M,d)$ where $M$ is a set and $d$ is a metric on $M$ , i.e., a function:-
$$
d : M \times M \to \mathbb{R}
$$
such that for $x,y,z \in M$ , the following holds:-
\begin{enumerate}
\item { $d(x,y) \ge 0 $ }
\item { $ d(x,y) = 0 $ iff $x=y$ . }
\item { $d(x,y) = d(y,x) $ }
\item { $d(x,z) \le  d(x,y) + d (y,z)$ }
\end{enumerate}
The function `$d$' is also called ``distance function'' or simply ``distance''.
\end{definition}

\begin{definition}\label{orbit}
\textbf{Orbit.}

Let $f:X \to X$ be a function. 
The sequence $ \mathcal{O} = \{x_0, x_1,x_2,x_3,...\}$ where
$$
x_{n+1} = f(x_n) \; ; \; x_n \in X \; ; \; n \ge 0 
$$
is called an orbit (more precisely `forward orbit') of $f$. 
\end{definition}

\begin{definition}\label{period}
\textbf{Period of An Orbit.}

The length of the orbit (definition \ref{orbit} ) is called the period 
of the orbit. Hence, for an orbit $ \mathcal{O}$ the period `$p$' is:-
$$
p = | \mathcal{O} | 
$$ 
$f$ is said to have a `closed' or `periodic' orbit $ \mathcal{O}$ if $| \mathcal{O}| \ne \infty$
or the period is not infinity.
\end{definition}

\begin{theorem}\label{cc}
\textbf{Convergence Criterion for a Fixed Point Iteration (Banach Fixed Point). }

Iteration of  a function $f:X \to X$ $x_{n+1} = f(x_n)$ would converge to 
the fixed point $x^*$ iff:-
$$
|f(x) - f(x^*)| \le |x - x^*|
$$
where $x,x^* \in A \subseteq X $ where $A$ is called the basin of attraction. 
\end{theorem}
That means, in short if $x^*$ be the fixed point, and $|f'(x)| \ge 1$ in 
the neighbourhood of the $x^*$ then the fixed point iteration 
won't converge at the fixed point $x^*$.  

\begin{definition}\label{top-space}
\textbf{Topological Space.}

Let the empty set be written as : $\emptyset$. Let $2^X$ denotes the power set, i.e. the set of all subsets of $X$.
A topological space is a set $X$ together with $\tau \subseteq 2^X$ satisfying the following axioms:-
\begin{enumerate}
\item{ $\emptyset \in \tau$ and $X \in \tau$ ,}
\item{ $\tau$ is closed under arbitrary union, }
\item{ $\tau$ is closed under finite intersection. }
\end{enumerate}
The set $\tau$ is called a topology of $X$.
\end{definition} 

\begin{definition}\label{dense-set}
\textbf{Dense Set.}

Let $A$ be a subset of a topological space $X$. 
$A$ is dense in $X$ for any point $x \in X$, if any neighborhood of $x$ contains at least one point from $A$.
\end{definition}

The real numbers $\mathbb{R}$ with the usual topology have the rational numbers $\mathbb{Q}$ as a countable dense subset.

\begin{definition}\label{top-trans}
\textbf{Topological Transitivity}

A function $f:X \to X$ is topologically transitive ,if, 
given any every pair of non empty open sets $A,B \subset X$ , 
there is some positive integer $n$ such that 
$$
f^n(A) \cap B \ne \emptyset .
$$ 
where $f^n$ means n'th iterate of $f$ .
\end{definition}
\end{section}

\end{document}